\documentclass[sigconf]{acmart}

\AtBeginDocument{%
  \providecommand\BibTeX{{%
    \normalfont B\kern-0.5em{\scshape i\kern-0.25em b}\kern-0.8em\TeX}}}

\usepackage{amsmath}
\usepackage{amsthm}
\newtheorem{rem}{remark}
\usepackage{algorithm}
\usepackage{algpseudocode}
\usepackage{textcomp}
\usepackage{multirow}

\begin{document}

\title{On Efficient Approximation of the Maximum Distance to A Point Over an Intersection of Balls}

\author{Beniamin Costandin}
\email{bcostandin@yahoo.com}
\affiliation{ \institution{Technical University of Cluj Napoca} \country{Romania}}

\author{Marius Costandin}
\email{costandinmarius@gmail.com}
\affiliation{ \institution{General Digits} \country{Romania}}

\begin{abstract}
In this paper we study the NP-Hard problem of maximizing the distance over an intersection of balls to a given point. We expand the results found in \cite{funcos1}, where the authors characterize the farthest in an intersection of balls $\mathcal{Q}$ to the given point $C_0$ by constructing some intersection of halfspaces. In this paper, by slightly modifying the technique found in literature, we characterize the farthest in an intersection of balls $\mathcal{Q}$ with another intersection of balls $\mathcal{Q}_1$. As such, going backwards, we are naturally able to find the given intersection of balls $\mathcal{Q}$ as the max indicator intersection of balls of another one $\mathcal{Q}_{-1}$. By repeating the process, we find a sequence of intersection of balls $(\mathcal{Q}_{i})_{i \in \mathbb{Z}}$, which has $\mathcal{Q}$ as an element, namely $\mathcal{Q}_{0}$ and show that $\mathcal{Q}_{-\infty} = \mathcal{B}(C_0,R_0)$ where $R_0$ is the maximum distance from $C_0$ to a point in $\mathcal{Q}$. As a final application of the proposed theory we give a polynomial algorithm for computing the maximum distance under an oracle which returns the volume of an intersection of balls, showing that the later is NP-Hard. Finally, we present a randomized method 
which allows an approximation of the maximum distance.
\end{abstract}

%



\maketitle

\section{Introduction}

In modern times, the systematic investigation of the geometry of the intersection of congruent balls (that is balls with equal radius) was started with the paper \cite{ballPoly1}. There are three books that survey some particular parts of the literature dealing with such intersections: \cite{bookballPoly1},  \cite{bookballPoly2}, and \cite{bookballPoly3}. For more general references, perhaps it is worth choosing from there.

In this paper we shall study the problem of maximizing the distance to a given point $C_0$ over an intersection of balls. We allow the balls to have arbitrary radius, but we still call sometimes their intersection as "ball-polyedra" in absence of a better term. This is a NP-Hard problem in general, although it allows a polynomial algorithm for some particular classes, as shown in \cite{funcos1} namely for the cases where $C_0$, the given point, is outside of the convex hull of the balls centers. 

However, to this date, the authors are not aware of any method which generally solves the problem if the point $C_0$ belongs to the convex hull of the balls centers. Very shortly, we show in the following that the Subset Sum Problem can be written as such a distance maximization problem, making this problem NP-Hard.

Indeed briefly, as presented in \cite{funcos1} let $n \in \mathbb{N}$ and consider $S \in \mathbb{R}^n$ and $T \in \mathbb{R}$. The associated subset sum problem, SSP(S,T) asks if exists $x \in \{0,1\}^{n}$ such that $x^T\cdot S = T$. For this, similar to \cite{sahni}, consider the optimization problem for $\beta > 0$:
\begin{align}\label{E3.13a}
& \max x^T\cdot(x - 1_{n \times 1}) + \beta \cdot S^T \cdot x \nonumber \\
& \hspace{0.5cm} \text{s.t} \ \ \   x \in \begin{cases} S^T\cdot x \leq T\\
0 \leq x_i \leq 1 \hspace{0.3cm} \forall i \in \{1, \hdots, n\}
\end{cases} 
\end{align} Let the feasible set be denoted by $\mathcal{P} = \{x \in \mathbb{R}^n| S^T\cdot x \leq T, 
0 \leq x_i \leq 1 \hspace{0.3cm} \forall i \in \{1, \hdots, n\} \}$. 
\begin{rem}\label{R1}
It is easy to see that the objective function is always smaller than or equal to $\beta \cdot T$. In fact the objective function reaches the value $\beta \cdot T$ if and only if the SSP(S,T) has a solution. 
\end{rem}

Note that the objective function can be rewritten as
\begin{align}\label{E2b}
&x^T\cdot x + \left(\beta \cdot S - 1_{n \times 1} \right)^T \cdot x = \nonumber \\
&= \left\| x - \frac{1_{n\times 1} - \beta \cdot S}{2} \right\|^2 - \left\| \frac{1_{n\times 1} - \beta \cdot S}{2} \right\|^2 \nonumber \\
& = \|x - C_0\|^2 - \|C_0\|^2
\end{align} with obvious definition for $C_0$. Since $C_0$ does not depend on $x$, we shall consider the optimization problem:
\begin{align}\label{E3}
\max_{x \in \mathcal{P}} \|x - C_0\|^2
\end{align}  Using Remark \ref{R1}, note that the SSP has a solution iff (\ref{E2b}) is zero, that is the mximum distance in (\ref{E3}) is $\|C_0\|^2$. The problem (\ref{E3}) is a distance maximization over a polytope.  Indeed $\mathcal{P}$ is the intersection of the unit hypercube with the halfspace $\{ x | S^T\cdot x \leq T\}$. Any maximizer shall be located in a corner of the polytope $\mathcal{P}$. Next, in \cite{funcos1} the polytope $\mathcal{P}$ is replaced with an intersection of balls which preserves the corners of the unit-hypercube: each hyperplane is replaced by a ball who's boundary leaves the same imprint on the boundary of $\mathcal{B}\left(\frac{1}{2}\cdot 1_{n \times 1}, \frac{\sqrt{n}}{2}\right)$ as the hyperplane defining the polytopes $\mathcal{P}$ facet. It is proven that if the SSP has a solution then this is also the maximizer of the maximum distance to $C_0$ over the proposed intersection of balls. From  \cite{funcos1}, it is worth noting that the centers of the balls in the constructed intersection of balls are required to have $C_0$ in their convex hull. 

It is worth noting that throughout this paper we shall denote by $\mathcal{B}(C,R)$ the open ball centered at $C\in \mathbb{R}^n$ with radius $R > 0$ and with $\bar{\mathcal{B}}(C,R)$ the closed ball centered at $C\in \mathbb{R}^n$ with radius $R > 0$. We also denote by $1_{n\times 1}$ the vector in $\mathbb{R}^n$ where all entries are $1$ and sometimes we refer to $\mathbb{R}^n$ as $\mathbb{R}^{n \times 1}$ for $n \in \mathbb{N}$. 

\section{Main Results}
This section contains several subsection through which the theory is presented. We start with a characterization of the farthest in an intersection of balls $\mathcal{Q}$ to a given point $C_0$. 
\subsection{Geometry Results: Characterization of Maximizers Over Intersection of Balls}

Let $m,n \in \mathbb{N}$ and $\bar{\mathcal{B}}(C_k, r_k) \subseteq \mathbb{R}^{n}$ be given closed balls for $k \in \{1, \hdots, m\}$. Let $\mathcal{Q}$ denote their intersection
\begin{align}
\mathcal{Q} = \bigcap_{k = 1}^m \bar{\mathcal{B}}(C_k, r_k)
\end{align} 

For a given $C_0 \in \mathbb{R}^{n }$ consider the problem:
\begin{align}\label{E2a}
R^{2}_0 = \max_{x \in \mathcal{Q}} \|x - C_0\|^2
\end{align}

In order to study the problem given by (\ref{E2a}) form the functions:
\begin{align} \label{E2}
h(x) = \max_{k \in \{1, \hdots, m\}} \|x - C_k\|^2 - r_k^2 \hspace{1cm} g_{\lambda}(x) = \lambda \cdot \|x - C_0\|^2
\end{align} for $ \lambda \in (0,1)$.  

\begin{rem}
Note that 
\begin{align}
\mathcal{Q} = \bigcap_{k = 1}^m \bar{\mathcal{B}}(C_k,r_k) = \{x | h(x) \leq 0\}
\end{align} which means that the intersection of balls is a sub-level set of the function $h$.
\end{rem}

\begin{rem}
Note that $h(x) - g_{\lambda}(x)$ is a convex function. Indeed
\begin{align}
&\|x - C_k\|^2 - r_k^2 - \lambda\cdot \|x - C_0\|^2 = \nonumber \\
& = \|x\|^2  -2\cdot x^T\cdot C_k + \|C_k^2\| - r_k^2 - \nonumber \\
& - \lambda \cdot (\|x\|^2-2\cdot x^T\cdot C_0 + \|C_0\|^2)  \nonumber \\
& = (1-\lambda)\cdot \|x\|^2 - 2\cdot x^T \cdot (C_k - \lambda \cdot C_0) + \|C_k\|^2 - \lambda\cdot \|C_0\|^2 - r_k^2 \nonumber \\
& = (1-\lambda)\cdot\left\| x - \frac{C_k - \lambda \cdot C_0}{1 - \lambda} \right\|^2 - \frac{\|C_k - \lambda\cdot C_0\|^2}{1 - \lambda} + \nonumber \\
& +  \|C_k\|^2 - \lambda\cdot \|C_0\|^2 - r_k^2 \nonumber \\
& =  (1-\lambda)\cdot\left\| x - \frac{C_k - \lambda \cdot C_0}{1 - \lambda} \right\|^2  - \frac{\lambda}{1 - \lambda} \cdot \|C_0 - C_k\|^2 - r_k^2
\end{align}
hence 
\begin{align}\label{E5f}
&h(x) - g_{\lambda}(x) = \nonumber \\
 &\max_{k \in \{1, \hdots, m\}}(1-\lambda)\cdot\left\| x - \frac{C_k - \lambda \cdot C_0}{1 - \lambda} \right\|^2  - \frac{\lambda}{1 - \lambda} \cdot \|C_0 - C_k\|^2 - r_k^2 
\end{align}
\end{rem}

Next, define a family of sets for $R \geq 0$
\begin{align}\label{E6f}
\mathcal{Q}_{R^2} = \{x | h(x) - g_{\lambda}(x) \leq -\lambda\cdot R^2\}
\end{align}. It is obvious that all the sets in the family are convex.

\begin{rem} Note that $\mathcal{Q}_{R^2}$ is actually an intersection of balls. Indeed $h(x) - g_{\lambda}(x) \leq -\lambda\cdot R^2$ is equivalent to 
\begin{align}
(1-\lambda)\cdot\left\| x - \frac{C_k - \lambda \cdot C_0}{1 - \lambda} \right\|^2  - \frac{\lambda}{1 - \lambda} \cdot \|C_0 - C_k\|^2 - r_k^2 \leq -\lambda\cdot R^2
\end{align} for all $k \in \{1, \hdots, m\}$. This is
\begin{align}\label{E9}
\left\| x - \frac{C_k - \lambda \cdot C_0}{1 - \lambda} \right\|^2 \leq \frac{1}{1-\lambda}\cdot \left( -\lambda\cdot R^2 + \frac{\lambda}{1 - \lambda} \cdot \|C_0 - C_k\|^2 + r_k^2 \right)
\end{align} for all $k \in \{1, \hdots, m\}$. 
\end{rem}

\begin{rem}\label{R3}
Note that $\mathcal{Q}$ is in $\mathcal{Q}_{0^2}$. Indeed, let $x_1$ belong to the boundary of $\mathcal{Q}$, then $h(x_1) = 0$, hence 
\begin{align}
h(x_1) - g_{\lambda}(x_1) = -g_{\lambda}(x_1) = -\lambda \cdot \|x_1 - C_0\|^2 \leq 0
\end{align} therefore $x_1 \in \mathcal{Q}_{0^2}$
\end{rem}

Having seen in Remark \ref{R3} what properties in relation to $\mathcal{Q}$ the set $\mathcal{Q}_{R^2}$ has for $R = 0$ (i.e $\mathcal{Q} \subseteq \mathcal{Q}_{0^2}$), it is natural to ask what happens if we increase $R$? Because $h,g_{\lambda}$ are bounded on bounded sets it is obvious that for large enough values of $R$ the set $\mathcal{Q}_{R^2}$ does not have elements in the fixed set $\mathcal{Q}$. Therefore the members of the family $\mathcal{Q}_{R^2}$ evolve as $R$ increases, from initially containing the set $\mathcal{Q}$ to not having elements in it. 
This is clarified in the following and this is the main idea of this section. 

In order to give the main result of this section, consider the convex optimzation problem:
\begin{align}\label{E8a}
\mathcal{H}^{\star} = \mathop{\text{argmin }}_{ h(x) \leq 1} h(x) - g_{\lambda}(x)
\end{align}
This is a convex optimization problem and can be solved in polynomial time.  
\begin{rem} \label{R4}
 The main observations are presented below:
\begin{enumerate}
\item Because $\mathcal{H}^{\star} \subseteq \{x | h(x) \leq 1\}$ follows that $\mathcal{H}^{\star}$ is bounded.
\item Since $\{x| h(x) \leq 1\}$ is bounded, exists $\infty > \bar{R} > 0$ such that $| h(x) - g_{\lambda}(x) | < \bar{R}^2$ for all $h(x) \leq 1$. Hence $\mathcal{Q}_{R^2} \cap \{x| h(x) \leq 1\} = \emptyset$ for all $R \geq \bar{R}$.  
\item Let $y \in \mathcal{H}^{\star}$. Then we denote
\begin{align}\label{E9a}
-\underline{R}^2 = h(y) - g_{\lambda}(y) \geq -\bar{R}^2 > -\infty
\end{align} It is easy to see that $\mathcal{H}^{\star} = \mathcal{Q}_{\underline{R}^2} \cap \{x | h(x) \leq 1\}$
\end{enumerate}
\end{rem}

The main theorem is presented in the following:
\begin{theorem} \label{T1}
With the notation from above the following alternatives are true:
\begin{enumerate}
\item If $\mathcal{H}^{\star} \subseteq \text{int} (\mathcal{Q}) $ then 
\begin{align}
R_0 = \max_{x \in \mathcal{Q}} \|x - C_0\| = \min \{R \geq 0| \mathcal{Q}_{R^2} \subseteq \mathcal{Q}\} 
\end{align}
\item If $\mathcal{H}^{\star} \subseteq \text{int} (\mathbb{R}^{n \times 1} \setminus \mathcal{Q}) $ then
\begin{align}
R_0 = \max_{x \in \mathcal{Q}} \|x - C_0\| = \max \{R \geq 0| \mathcal{Q}_{R^2} \cap \mathcal{Q} \neq \emptyset \} 
\end{align}
\item If $\mathcal{H}^{\star} \cap \partial\mathcal{Q} \neq \emptyset $ then $ \forall y \in \mathcal{H}^{\star}  \cap \partial \mathcal{Q}$ one has
\begin{align}
\underline{R}^2 \leq R_0^2 = \max_{x \in \mathcal{Q}} \|x - C_0\|^2 \leq \frac{1}{\lambda} \cdot \underline{R}^2 =  \frac{1}{\lambda} \cdot |h(y) - g_{\lambda}(y)| 
\end{align} 
\end{enumerate}
\end{theorem}

\begin{proof}
First recall that $\mathcal{Q} \subseteq \mathcal{Q}_{0^2}$ from Remark \ref{R3} and let $\underline{R}$ be defined by (\ref{E9a}). 
\begin{enumerate}
\item Proof for the case $\mathcal{H}^{\star} \subseteq \text{int} (\mathcal{Q}) $ : 
 From Remark \ref{R4} statement 3, follows that $\mathcal{Q}_{\underline{R}^2} \subseteq \mathcal{Q} \subseteq \mathcal{Q}_{0^2}$ hence the set $\{R >0 | \mathcal{Q}_{R^2} \subseteq \mathcal{Q}\} \neq \emptyset$. Furthermore, because $\mathcal{Q} \subseteq \mathcal{Q}_{0^2}$ follows that the minimum of this set is strictly positive. Let 
\begin{align}
\tilde{R} = \min \{R > 0| \mathcal{Q}_{R^2} \subseteq \mathcal{Q}\} 
\end{align} and it is proven that: 
\begin{enumerate}
\item $R_0 \leq \tilde{R}$ First, note from the definition of $\tilde{R}$ that $\mathcal{Q}_{R^2} \cap \partial \mathcal{Q} = \emptyset $ for all $R> \tilde{R}$ (since $\partial \mathcal{Q}_{R^2} \cap \partial \mathcal{Q}_{T^2} = \emptyset$ for all $R \neq T$ and $\mathcal{Q}_{R^2} \subseteq \mathcal{Q}_{T^2} $ for all $R \geq T$ ). 

 Next, let $ z \in \partial \mathcal{Q}$ with $\|z - C_0\| = R_0$. It follows
\begin{align}
h(z) - g_{\lambda}(z) = 0 - \lambda \cdot R_0^2  \hspace{0.2cm} \Rightarrow \hspace{0.2cm} z \in \mathcal{Q}_{R_0^2}
\end{align} hence $z \in \mathcal{Q}_{R_0^2} \cap \partial \mathcal{Q}$. Assuming  that $R_0 > \tilde{R}$, a contradiction is obtained with $\mathcal{Q}_{R_0^2} \cap \partial \mathcal{Q} = \emptyset$. 

\item $R_0 \geq \tilde{R}$ Indeed, assume that $R_0 < \tilde{R}$ and let $z \in \partial \mathcal{Q} \cap \mathcal{Q}_{\tilde{R}^2}$ then 
\begin{align}
h(z) - g_{\lambda}(z) \leq - \lambda\cdot \tilde{R}^2 < -\lambda\cdot R_0^2 \hspace{0.2cm} \Rightarrow \hspace{0.2cm} \|z - C_0\|^2 > R_0^2
\end{align} which is a contradiction with the definition of $R_0$
\end{enumerate}
\item Proof for the case $\mathcal{H}^{\star} \subseteq \text{int} (\mathbb{R}^{n \times 1} \setminus \mathcal{Q}) $ : Recall from Remark \ref{R4} statement 3 that 
\begin{align}
\mathcal{H}^{\star} = \mathcal{Q}_{\underline{R}^2} \cap \{x | h(x) \leq 1\}
\end{align} therefore in this case  one has $\mathcal{Q}_{\underline{R}^2} \cap \mathcal{Q} = \emptyset$. Since $\mathcal{Q} \subseteq \mathcal{Q}_{0^2}$ follows that the set $\{R > 0 | \mathcal{Q} \cap \mathcal{Q}_{R^2} \neq \emptyset \}$ is not empty and is bounded. Let 
\begin{align}
\tilde{R} = \max \{R > 0 | \mathcal{Q} \cap \mathcal{Q}_{R^2} \neq \emptyset \}
\end{align} and one proves that
\begin{enumerate}
\item $R_0 \leq \tilde{R}$ First, note from the definition of $\tilde{R}$ that $\mathcal{Q}_{R^2} \cap \partial \mathcal{Q} = \emptyset $ for all $R> \tilde{R}$ (it is the same statement with the previous case in the theorem, but here is a different reasons for its validity. In this case if $\exists R > \tilde{R}$ with $\partial \mathcal{Q} \cap \mathcal{Q}_{R^2} \neq \emptyset$ then $\mathcal{Q} \cap \mathcal{Q}_{R^2} \neq \emptyset$ and this contradicts the definition of $\tilde{R}$). 

 Next, let $ z \in \partial \mathcal{Q}$ with $\|z - C_0\| = R_0$. It follows
\begin{align}
h(z) - g_{\lambda}(z) = 0 - \lambda\cdot R_0^2  \hspace{0.2cm} \Rightarrow \hspace{0.2cm} z \in \mathcal{Q}_{R_0^2}
\end{align} hence $z \in \mathcal{Q}_{R_0^2} \cap \partial \mathcal{Q}$. Assuming  that $R_0 > \tilde{R}$, one obtains a contradiction with $\mathcal{Q}_{R_0^2} \cap \partial \mathcal{Q} = \emptyset$.

\item $R_0 \geq \tilde{R}$ It is know that $\mathcal{Q} \cap \mathcal{Q}_{\tilde{R}^2} \neq \emptyset$ and $\mathcal{Q} = \text{int}(\mathcal{Q}) \cup \partial \mathcal{Q}$ and since it can be argued that $\text{int}(\mathcal{Q}) \cap \mathcal{Q}_{\tilde{R}^2} = \emptyset$ follows that $\partial \mathcal{Q} \cap \mathcal{Q}_{\tilde{R}^2} \neq \emptyset$. 

Assume that $R_0 < \tilde{R}$ and let $z \in \partial \mathcal{Q} \cap \mathcal{Q}_{\tilde{R}^2}$  then 
\begin{align}
h(z) - g_{\lambda}(z) \leq -\lambda\cdot \tilde{R}^2 < -\lambda\cdot R_0^2 \hspace{0.2cm} \Rightarrow \hspace{0.2cm} \|z - C_0\|^2 > R_0^2
\end{align} which is a contradiction with the definition of $R_0$
\end{enumerate}

\item Proof for the case $\mathcal{H}^{\star} \cap \partial\mathcal{Q} \neq \emptyset $ : 
 One proves that 
\begin{enumerate}
\item $R_0 \leq \frac{1}{\sqrt{\lambda}}\cdot \underline{R}$ Indeed, let $z \in \partial \mathcal{Q} \subseteq \{x | h(x) \leq 1\}$ with $\| z - C_0\| = R_0$ and assume that $R_0 > \frac{1}{\sqrt{\lambda}}\cdot \underline{R}$. Then 
\begin{align}
h(z) - g_{\lambda}(z) = 0 - \lambda\cdot R_0^2 < - \underline{R}^2 
\end{align} contradicting the definition of $\underline{R}$ from Remark \ref{R4} st. 3 as the minimum value of $h(x) - g_{\lambda}(x)$ over the set $\{x| h(x) \leq 1\}$. 

\item $R_0 \geq\underline{R}$ From Remark \ref{R4} statement 3 follows that 
\begin{align}
\mathcal{H}^{\star} = \mathcal{Q}_{\underline{R}^2} \cap \{x | h(x) \leq 1\} \hspace{0.2cm} \Rightarrow \hspace{0.2cm} \mathcal{Q}_{\underline{R}^2} \cap \{x | h(x) \leq 1\} \cap \partial \mathcal{Q} \neq \emptyset
\end{align} Assume that $R_0 < \underline{R}$ and let $z \in \partial \mathcal{Q} \cap \mathcal{Q}_{\underline{R}^2}$. It follows
\begin{align}
h(z) - g_{\lambda}(z) = 0 - \lambda\cdot  \|z - C_0\|^2 \leq - \lambda\cdot \underline{R}^2 \hspace{0.2cm} \Rightarrow \hspace{0.2cm} \|z - C_0\| \geq \underline{R} > R_0
\end{align} which is a contradiction with the definition of $R_0$
\end{enumerate}
\end{enumerate}
\end{proof}

Similarly to \cite{funcos1} Corollary 1, the following are true for this case as well: 

\begin{corollary}\label{C2} With the notations from above, if \\ $C_0 \not\in \text{conv}\{C_1, \hdots, C_m\}$ then $\mathcal{H}^{\star} \subseteq \text{int}(\mathbb{R}^{n \times 1} \setminus \mathcal{Q})$  
\end{corollary}

\begin{proof}
Begin by stating that since $C_0 \not\in \text{conv}\{C_1, \hdots, C_m\}$ then exists a hyperplane $\{x | A^T \cdot x + b = 0\}$ with $A \in \mathbb{R}^{n \times 1}$ and $b \in \mathbb{R}$ such that 
\begin{align}\label{E12d}
A^T \cdot C_0 + b< 0 \hspace{0.5cm} A^T \cdot C_k + b > 0 \hspace{0.5cm} \forall k \in \{1, \hdots, m\}
\end{align}.  

 One now shows that $\mathcal{Q} \cap \mathcal{H}^{\star} = \emptyset$. Indeed, assuming otherwise, let $y^{\star} \in \mathcal{H}^{\star} \cap \mathcal{Q}$ and let 
$v = \frac{A}{\|A\|}$.
We show that $\exists \alpha > 0$ such that $y = y^{\star} + \alpha \cdot v \in \{x | h(x) \leq 1\}$ and $h(y) - g_{\lambda}(y) < h(y^{\star}) - g_{\lambda}(y^{\star})$ contradicting the fact that $y^{\star}$ is a minimum (of $h - g_{\lambda}$). 

First, since $y^{\star} \in \mathcal{Q} = \{x | h(x) \leq 0\}$ acknowledge the existence of $\alpha_1 > 0$ such that $y^{\star} + \alpha \cdot v \in \{x | h(x) \leq 1\}$ for all $0 < \alpha \leq \alpha_1$ 

Next, assume w.l.o.g that $h(y^{\star}) = \|y^{\star} - C_k\|^2 - r_k^2$ for all $k \in \{1, \hdots, p \leq m\}$.  Then exists $\alpha_2 > 0$ such that $h(y^{\star} + \alpha \cdot v) = \| y^{\star} + \alpha \cdot v - C_k\|^2 - r_k^2$ for some $k \in \{1, \hdots, p \leq m\}$ for all $0 < \alpha \leq \alpha_2$.

Let $\alpha_0 = \min \{\alpha_1, \alpha_2\}$ and $0 < \alpha \leq \alpha_0$ then

\begin{align}
h(y^{\star} + \alpha\cdot v) - g_{\lambda}(y^{\star} + \alpha \cdot v) &< h(y^{\star}) - g_{\lambda}(y^{\star}) 
\end{align}

 Indeed, since as stated above $\exists k \in \{1, \hdots, m\}$ such that $h(y^{\star}) = \|y^{\star} - C_k\|^2 - r_k^2$ and $h(y^{\star} + \alpha \cdot v) = \|y^{\star} + \alpha \cdot v - C_k\|^2 - r_k^2$ follows that 

\begin{align}
&\| y^{\star} + \alpha \cdot v - C_k\|^2 - r_k^2 - \lambda\cdot \|y^{\star} + \alpha \cdot v - C_0\|^2 <  \nonumber \\
&<\| y^{\star} - C_k\|^2 - r_k^2 - \lambda\cdot \|y^{\star} - C_0\|^2
\end{align} is equivalent to 
\begin{align} 
&\| y^{\star} + \alpha \cdot v - C_k\|^2 - \|y^{\star} - C_k\|^2 < \nonumber \\
&< \lambda\cdot \left( \|y^{\star} + \alpha \cdot v - C_0\|^2 - \|y^{\star} - C_0\|^2\right)  
\end{align} which is 
\begin{align}
&\alpha \cdot v^T \cdot \left( 2 (y^{\star} - C_k) + \alpha \cdot v\right) < \lambda\cdot \alpha \cdot v^T \cdot \left( 2 (y^{\star} - C_0) + \alpha \cdot v\right) \iff \nonumber \\
&2 \alpha \cdot v^T \cdot (y^{\star} - C_k) < 2 \alpha \cdot v^T \cdot (y^{\star} - C_0) + (\lambda - 1)\cdot \alpha^2\cdot \|v\|^2  \iff \nonumber \\
& (1- \lambda)\cdot \|\alpha\cdot v \|^2 + 2 \alpha \cdot v^T \cdot (C_0 - C_k) < 0 \iff \nonumber \\
&(1- \lambda)\cdot \alpha + 2\cdot v^T\cdot (C_0 - C_k) < 0
\end{align} Since $ v^T\cdot (C_0 - C_k) < 0$ follows that $\exists \alpha >0$ such that  the above are met.

\end{proof}

In \cite{funcos1} the authors provide a polynomial algorithm for finding the maximizer (proved to be unique) for the case in which $C_0 \not\in \text{conv}\{C_1, \hdots, C_m\}$ 

 These cases are therefore no longer of interest for us. This paper shall focus on the cases in which $C_0 \in \text{conv}\{C_1, \hdots, C_m\}$.  
In this situation therefore one has one of the two:
\begin{enumerate}
\item $\mathcal{H}^{\star} \subseteq \text{int}(Q)$
\item $\mathcal{H}^{\star} \cap \partial \mathcal{Q} \neq \emptyset$
\end{enumerate}

In order to distinguish between the two, one can simply compute $\mathcal{H}^{\star}$. Note that because $h - g_{\lambda}$ is a piecewise quadratic function, hence strictly convex and therefore its minimum is unique. 

Let $\{ y^{\star} \} = \mathcal{H}^{\star}$  and assume that $C_0 \in \text{conv}\{C_1, \hdots, C_m\}$. If $y^{\star} \in \partial \mathcal{Q}$ then we stop and according to Theorem \ref{T1}, return the result:
\begin{align}
R_0 = \max_{x \in \mathcal{Q}}\| x -C_0\| \in \left[\|y^{\star} - C_0\|, \frac{1}{\sqrt{\lambda}} \cdot \|y^{\star} - C_0\| \right]
\end{align}, otherwise, continue with the next subsection.

\subsection{Maximizing Distances Over Intersection of Balls if $\mathcal{H}^{\star} \subseteq \text{int}(\mathcal{Q})$}
 In the following we assume that $C_0 \in \text{int}(\text{conv}\{C_1, \hdots, C_m\})$ and apply Theorem \ref{T1} to solve $\max_{x \in \mathcal{Q}} \|x - C_0\|$. For this subsection, see Figure \ref{fig1} and \ref{fig2}. Let 
\begin{align}\label{E32a}
R_0= \max_{x \in \mathcal{Q}} \|x - C_0\|
\end{align} We assume that an interval for $R_0$ is known apriori, that is one knows $\underline{R}_0, \overline{R}_0 \in \mathbb{R}_+$ such that $\underline{R}_0 \leq R_0 \leq \overline{R}_0$. 

Theorem \ref{T1} assures the existence of the set $\mathcal{Q}_{R_0^2}$ which is an intersection of balls with the following properties:
\begin{enumerate}
\item $\mathcal{Q}_{R_0^2} \subseteq \mathcal{Q}$
\item the vertices of $\mathcal{Q}_{R_0^2}$ on the boundary of $\mathcal{Q}$ are the farthest points in $\mathcal{Q}$ to $C_0$.
\end{enumerate}

This is sufficient to assert that 
\begin{align}
R_0 = \max_{x \in \mathcal{Q}_{R_0^2}} \|x - C_0\|
\end{align} From (\ref{E9}), the balls forming $\mathcal{Q}_{R_0^2}$ are:

\begin{align}
\left\| x - \frac{C_k - \lambda \cdot C_0}{1 - \lambda} \right\|^2 \leq \frac{1}{1-\lambda}\cdot \left( -\lambda\cdot R_0^2 + \frac{\lambda}{1 - \lambda} \cdot \|C_0 - C_k\|^2 + r_k^2 \right)
\end{align} We note:
\begin{enumerate}
\item The centers of the new intersection of balls are:
\begin{align}\label{E35}
C_{k,1} :=  \frac{C_k - \lambda \cdot C_0}{1 - \lambda}
\end{align}
\item The radii of the new intersection of balls are given by:
\begin{align}\label{E36}
r_{k,1}^2 := \frac{1}{1-\lambda}\cdot \left( -\lambda\cdot R_0^2 + \frac{\lambda}{1 - \lambda} \cdot \|C_0 - C_k\|^2 + r_k^2 \right)
\end{align}
\item The centers of the balls in the new intersection of balls, i.e. $\mathcal{Q}_{R_0^2}$, do not depend on $R_0$.
\end{enumerate}

Let's denote $\mathcal{Q}_{R_0^2}^1 := \mathcal{Q}_{R_0^2}$ and hence $R_0 = \max_{x \in \mathcal{Q}_{R_0^2}^1} \|x - C_0\|$. The superscript will count the generations of balls centers, as seen below.

Proceed as follows:
\begin{enumerate}
\item If $C_0 \not\in \text{conv}\{C_{1,1}, \hdots, C_{m,1}\}$ then for any $R>0$ one can solve $\max_{\mathcal{Q}_{R^2}^1} \|x - C_0\|$ with the the polynomial algorithm given in \cite{funcos1}, to find $R_0 = \max_{\mathcal{Q}_{R_0^2}^1} \|x - C_0\|$. In this case, the algorithm stops and return $R_0$ as the solution to the optimization problem (\ref{E32a}). Note that $R_0$ is a fixed point of a uni-variate function $f(R) = \max_{\mathcal{Q}_{R^2}^1} \|x - C_0\|$. 

\item If $C_0 \in \text{conv}\{C_{1,1}, \hdots, C_{m,1}\}$ compute $\{y_1^{\star}\} = \mathcal{H}^{\star}_1$ associated with the new intersection of balls, $\mathcal{Q}_{R_0^2}^1$. We distinguish two cases
\begin{enumerate}
\item $y_1^{\star} \in \partial\mathcal{Q}_{R_0^2}^1 $ :  Since $ \partial \mathcal{Q}_{R_0^2}^1$ has zero measure in $\mathbb{R}^n$ and because we do not know how to assert it (because $R_0$ is not known) we shall ignore this case for the time being. 

\item $y_1^{\star} \in \text{int}(\mathcal{Q}_{R_0^2}^1) $ : We always consider this case. From Theorem \ref{T1} first statement, one concludes that exists $\mathcal{Q}_{R_0^2}^2$ such that $R_0 = \max_{x \in \mathcal{Q}_{R_0^2}^2} \|x - C_0\|$. From (\ref{E9}), the balls forming $\mathcal{Q}_{R_0^2}^2$ are:

\begin{align}
\left\| x - \frac{C_{k,1} - \lambda \cdot C_0}{1 - \lambda} \right\|^2 \leq \frac{1}{1-\lambda}\cdot \left( -\lambda\cdot R_0^2 + \frac{\lambda}{1 - \lambda} \cdot \|C_0 - C_{k,1}\|^2 + r_{k,1}^2 \right)
\end{align} 
The centers of the new intersection of balls are:
\begin{align}
C_{k,2} :=  \frac{C_{k,1} - \lambda \cdot C_0}{1 - \lambda}
\end{align} Note that these do not depend on $R_0$. 

The radii of the new intersection of balls are given by:
\begin{align}
r_{k,2}^2 := \frac{1}{1-\lambda}\cdot \left( -\lambda\cdot R_0^2 + \frac{\lambda}{1 - \lambda} \cdot \|C_0 - C_{k,1}\|^2 + r_{k,1}^2 \right)
\end{align}

\end{enumerate}
\end{enumerate} 

As such, we define the algorithm:

\begin{algorithm} 
\caption{Procedure A}\label{procA}
 This procedure given an intersection of balls $\mathcal{Q} = \bigcap_{k=1}^m \bar{\mathcal{B}}(C_k,r_k)$ computes a sequence of centers of balls. 
\begin{algorithmic}[1]
\Require $C_0$ , $C_{1}, \hdots, C_{m}$ and $r_1, \hdots, r_m$
 
\State $C_{k,0} \gets C_k$
\State $i \gets 0$
\While{$C_0 \in \text{conv}\{C_{1,i}, \hdots, C_{m,i}$ \}}
\State $C_{k,i+1} :=  \frac{C_{k,i} - \lambda \cdot C_0}{1 - \lambda}$
\State $i \gets i+1$
\EndWhile

\end{algorithmic}

\end{algorithm}

 By repeating the process, we are therefore able to assert the existence of a sequence of intersection of balls $\mathcal{Q}_{R_0^2}^i$ with the following properties:
\begin{align}\label{E40}
\mathcal{Q}_{R_0^2}^{i+1} &\subseteq \mathcal{Q}_{R^2_0}^i \nonumber \\
\max_{x \in \mathcal{Q}_{R_0^2}^i} \|x - C_0\| = &\hdots = \max_{x \in \mathcal{Q}_{R_0^2}^1} \|x - C_0\| = R_0
\end{align}

As explained above, the presented process either stops, returning $R_0$ (if $C_0$ remains outside the convex hull of the centers of some generated intersection of balls $\mathcal{Q}_{R_0^2}^i$) or continues by returning the centers of a new intersection of balls ( if $C_0$ remains in the convex hull of the centers of the generated intersection of balls $\mathcal{Q}_{R_0^2}^i$ ). 

Note that Procedure A from Algorithm \ref{procA}, in case $C_0$ is in the convex hull of the centers of the intersection of balls, is able to generate the centers of the next intersection of balls even though the balls radii are not known, because as noted, these do not depend on $R_0$. 

 One naturally asks: does this process stop? That is, does at any iteration $i \geq 1$, the point $C_0$ remain outside the convex hull of the balls centers of the intersection of balls $\mathcal{Q}_{R_0^2}^i$ ? The answer is NO in general, as can be verified with some immediate examples. For this see Figure \ref{fig3} and \ref{fig4}.

Our approach to this situation, is the main result of this paper: we ask, for a given intersection of ball $\mathcal{Q}$ with balls centers in $C_k$ and of radius $r_k$ with the given point $C_0$: \textit{ is this intersection of ball generated from an intersection of balls which had $C_0$ outside the convex hull of the centers of the balls forming it}? That is, instead of going outwards, with forming new intersection of balls on top of what is given, and as such leaving $C_0$ deeper and deeper in the convex hull of the newly generated balls centers, can we go inwards? We will call that intersection of balls a "seed" (out of which the given intersection of balls grew, through the explained process in Procedure A Algorithm \ref{procA}). 

\subsection{Analysis of the reverse sequence}

From \ref{E35} it is readily obvious that denoting $C_{k,-1}$ the centers of the previous generation intersection of balls, one gets:
\begin{align}
C_k = \frac{ C_{k,-1} - \lambda \cdot C_0}{1 - \lambda} \hspace{0.1cm}\Rightarrow \hspace{0.1cm} C_{k,-1} = (1-\lambda)\cdot C_k  + \lambda \cdot C_0
\end{align} repeating the process one gets:

\begin{align}\label{E42}
C_{k,-i} = \frac{ C_{k,-(i+1)} - \lambda \cdot C_0}{1 - \lambda} \hspace{0.1cm}\Rightarrow \hspace{0.1cm} C_{k,-(i+1)} = (1-\lambda)\cdot C_{k,-i}  + \lambda \cdot C_0
\end{align} then

\begin{align}\label{E46c}
C_{k,-i} &= (1-\lambda)^i \cdot C_{k,0} + \lambda \cdot C_0 \cdot \sum_{p=0}^{i-1} (1-\lambda)^p \nonumber \\
& = (1-\lambda)^i \cdot C_{k,0} + \left( 1 - (1 - \lambda)^i \right)\cdot C_0
\end{align} hence $\|C_0- C_{k,-i-1}\| = (1-\lambda)\cdot \|C_0 - C_{k,-i}\| = \hdots = (1-\lambda)^i \cdot \|C_0 - C_{k,-1}\|$ therefore $C_{k,-i} \to C_0$.
%
%


%

Since the existence of the centers of the balls has been positively established, one focuses on the existence of radii, $r_{k,-i}$ of the balls centered in $C_{k,-i}$ such that after the application of the above presented process the given radii, $r_k$ are obtained.  Starting with the radius $r_{k,-i}$, from (\ref{E36}) one gets

\begin{align}\label{E43}
&r_{k,-i+1}^2 := \frac{1}{1-\lambda}\cdot \left( -\lambda\cdot R_0^2 + \frac{\lambda}{1 - \lambda} \cdot \|C_0 - C_{k,-i}\|^2 + r_{k, -i}^2 \right) \nonumber \\
& r_{k,-i+2}^2 := \frac{1}{1-\lambda}\cdot \left( -\lambda\cdot R_0^2 + \frac{\lambda}{1 - \lambda} \cdot \|C_0 - C_{k,-i+1}\|^2 + r_{k, -i+1}^2 \right) \nonumber \\ 
& \vdots \nonumber \\
& r_{k}^2 := \frac{1}{1-\lambda}\cdot \left( -\lambda\cdot R_0^2 + \frac{\lambda}{1 - \lambda} \cdot \|C_0 - C_{k,-1}\|^2 + r_{k, -1}^2 \right) \nonumber \\
\end{align}

Using (\ref{E43}), because $C_{k,-j}$ are known for $j \in \{1, \hdots, i\}$ it is possible to start from the given $r_k$ and iteratively compute back $r_{k,-1}, r_{k,-2}, \hdots, r_{k,-i}$ as a function of $R_0 = \max_{x \in \mathcal{Q}} \|x - C_0\|$. Then similarily to (\ref{E40}) one has

\begin{align}
&\mathcal{Q} \subseteq \mathcal{Q}_{R_0^2}^{-1}  \subseteq \hdots \subseteq \mathcal{Q}_{R_0^2}^{-i+1} \subseteq \mathcal{Q}_{R^2_0}^{-i} \nonumber \\
&R_0 = \max_{x \in \mathcal{Q}_{R_0^2}^{-1}} \|x - C_0\| = \hdots = \max_{x \in \mathcal{Q}_{R_0^2}^{-i}} \|x - C_0\|
\end{align}

\begin{rem}\label{R6} It is interesting to note here an overview: since for $i \to \infty$ one has $C_{k,-i} \to C_0$  from (\ref{E43}) follows that $\min_k r_{k,-i} \to^{i \to \infty} R_0$. As such, basically  $\mathcal{Q}_{R_0^2}^{-i} \to^{i \to \infty} \bar{\mathcal{B}}(C_0,R_0)$. Even more, we will show below that $r_{k,-i} \to^{\infty} R_0$. 
\end{rem} For the above remark, see Figure \ref{fig5} and \ref{fig6}.

In the following, using (\ref{E43}) we write $r_{k,0} := r_{k}$ as a function of $R_0$ and $r_{k,-i}$. We also denote by $C_{k,0}:= C_k$. 
\begin{align}\label{E46}
r_{k,-i+1}^2 &= \frac{r_{k,-i}^2}{1-\lambda} + \frac{-\lambda\cdot R_0^2}{1-\lambda} +  \frac{\lambda}{(1-\lambda)^2} \cdot \|C_0 - C_{k,-i}\|^2 \nonumber \\
r_{k,-i+2}^2 &= \frac{r_{k,-i}^2}{(1-\lambda)^2} + \sum_{p=0}^1 \frac{-\lambda\cdot R_0^2}{(1-\lambda)^{(1+p)}} + \nonumber \\
& +  \frac{\lambda}{(1-\lambda)^3}\cdot \|C_0 - C_{k,-i}\|^2 + \frac{\lambda}{(1-\lambda)^2} \cdot \|C_0 - C_{k,-i+1}\|^2 \nonumber \\
 = \frac{r_{k,-i}^2}{(1-\lambda)^2} &+ \sum_{p=0}^1 \frac{-\lambda\cdot R_0^2}{(1-\lambda)^{(1+p)}} + \sum_{p=0}^1 \frac{\lambda}{(1-\lambda)^{(1 + 2-p)}}\cdot \|C_0 - C_{k,-i+p}\|^2 \nonumber \\
r_{k,-i+3}^2 &= \frac{r_{k,-i}^2}{(1-\lambda)^3} + \frac{1}{1-\lambda}\cdot \sum_{p=0}^1 \frac{-\lambda\cdot R_0^2}{(1-\lambda)^{(1+p)}} + \frac{-\lambda\cdot R_0^2}{1-\lambda} + \nonumber \\
+ \frac{1}{1-\lambda} &\cdot \sum_{p=0}^1 \frac{\lambda}{(1-\lambda)^{(1 + 2-p)}} \|C_0 - C_{k,-i+p}\|^2 + \frac{\lambda\cdot \|C_0 - C_{-i+2}\|^2}{(1-\lambda)^2}  \nonumber \\
 &= \frac{r_{k,-i}^2}{(1-\lambda)^3} + \sum_{p=0}^2 \frac{-\lambda\cdot R_0^2}{(1-\lambda)^{(1+p)}} + \sum_{p=0}^2 \frac{\lambda\cdot \|C_0 - C_{k,-i+p}\|^2}{(1-\lambda)^{(1 + 3-p)}}  \nonumber \\
&\vdots \nonumber \\
r_{k,-i+q+1}^2 &= \frac{r_{k,-i}^2}{(1-\lambda)^{(q+1)}} + \sum_{p=0}^q \frac{-\lambda\cdot R_0^2}{(1-\lambda)^{(1+p)}} +\sum_{p=0}^q \frac{\lambda\cdot \|C_0 - C_{k,-i+p}\|^2}{(1-\lambda)^{(1 +(q+1)-p)}} \nonumber \\
&\vdots \\
r_{k,0}^2 &= \frac{r_{k,-i}^2}{(1-\lambda)^{i}} + \sum_{p=0}^{i-1} \frac{-\lambda\cdot R_0^2}{(1-\lambda)^{(1+p)}} + \sum_{p=0}^{i-1} \frac{\lambda \cdot \|C_0 - C_{k,-i+p}\|^2 }{(1-\lambda)^{(1 + i-p)}} 
\end{align} From (\ref{E42}) one gets

\begin{align}
\|C_0 - C_{-i+p}\| = (1-\lambda)^{(i-p)} \cdot \|C_0 - C_{k,0}\|
\end{align} hence 

\begin{align}\label{E48}
\sum_{p=0}^{i-1} \frac{\lambda \cdot \|C_0 - C_{k,-i+p}\|^2 }{(1-\lambda)^{(1 + (i-1)-p)}} &= \frac{\lambda}{1 - \lambda} \cdot \|C_0 - C_{k,0}\|^2 \cdot \sum_{p=0}^{i-1} (1-\lambda)^{i-p} \nonumber \\
& = \frac{\lambda}{1-\lambda} \cdot \|C_0 - C_{k,0}\|^2 \cdot \left( \frac{1 - (1-\lambda)^{i+1}}{1 - (1-\lambda)} - 1\right) \nonumber \\
& = \left(1 - (1-\lambda)^i \right)\cdot \|C_0 - C_{k,0}\|^2
\end{align} Finally, because 
\begin{align}\label{E49}
 \sum_{p=0}^{i-1} \frac{-\lambda\cdot R_0^2}{(1-\lambda)^{(1+p)}} &= -\frac{\lambda}{1-\lambda} \cdot R_0^2 \cdot \sum_{p=0}^{i-1} \frac{1}{(1-\lambda)^p} \nonumber \\
& = -\frac{\lambda}{1-\lambda} \cdot R_0^2 \cdot  \left( \frac{1 - \frac{1}{(1-\lambda)^i}}{1 - \frac{1}{1-\lambda}}\right) \nonumber \\
&= \left(1 - \frac{1}{(1-\lambda)^i} \right)\cdot R_0^2
\end{align}

From (\ref{E46}, \ref{E48}, \ref{E49}) one gets

\begin{align}\label{E50}
r_{k,0}^2 = \frac{r_{k,-i}^2 - R_0^2}{(1-\lambda)^i} + R_0^2 + \left(1- (1-\lambda)^i \right)\cdot \|C_0 - C_{k,0}\|^2
\end{align}

\begin{rem} As stipulated in Remark \ref{R6} one gets from (\ref{E50}) that
\begin{align}\label{E54c}
r_{k,-i}^2 - R_0^2 &= (1 - \lambda)^i\cdot \left(r_{k,0}^2 - R_0^2 - \left(1 - (1-\lambda)^i \right) \cdot \|C_0 - C_{k,0}\|^2 \right) \nonumber \\
& \to^{i \to \infty} 0
\end{align} which indicates that the convergence of $r_{k,-i}$ to $R_0$ is exponential.
\end{rem}

Actually, one can take any $R>0$ and compute as such the sequence of intersection of balls for any $\lambda \in (0,1)$:

\begin{align}
\hdots, \mathcal{Q}_{R^2}^{i}, \mathcal{Q}_{R^2}^{i-1},  \hdots, \mathcal{Q}_{R^2}^{1}, \mathcal{Q}_{R^2}^0 = \mathcal{Q}, \mathcal{Q}_{R^2}^{-1}, \hdots, \mathcal{Q}_{R^2}^{i-1}, \mathcal{Q}_{R^2}^{i}, \hdots
\end{align} It is obtained that $\mathcal{Q}_{R^2}^{-i} \to^{i\to \infty} \bar{\mathcal{B}}(C_0,R)$. Furthermore $\mathcal{Q}_{R^2}^i \to^{i\to \infty} \mathcal{P}_{R^2}$ where $\mathcal{P}_{R^2}$ is a polytope. For $R = R_0 = \max_{x \in \mathcal{Q}} \|x - C_0\|$ one has
\begin{enumerate}
\item $\mathcal{Q}_{R_0^2}^{i} \subseteq \mathcal{Q}_{R_0^2}^j$ for all $i,j \in \mathbb{Z}$ with $i \geq j$
\item $R_0 = \max_{x \in \mathcal{Q}_{R_0^2}^i} \|x - C_0\| = \max_{x \in \mathcal{Q}_{R_0^2}^j} \|x - C_0\|$ for all $i,j \in \mathbb{Z}$. 
\end{enumerate}

\begin{rem} \label{R9}
For any $\lambda \in (0,1)$, let us denote $\mathcal{S}_{\lambda,R} := \left( \mathcal{Q}_{R^2}^i\right)_{i \in \mathbb{Z}}$ the sequence of intersection of balls. For any given $R > 0$, it is obvious that $\mathcal{Q}_{R}^0 := \mathcal{Q}$ hence one member of the sequence is known. If any other member is known, then one can find the whole sequence (of intersection of balls). That is, finding the maximum distance $R_0$ to $C_0$ over $\mathcal{Q}$ is a particular case of finding another member if the sequence $\mathcal{S}_{\lambda,R_0}$ i.e $\mathcal{Q}_{R_0^2}^{\infty}$. 

For a given $R >0$, as a future work on this area, one might investigate the validity of the statement:
\begin{align}\label{E57c}
&R < R_0 \Rightarrow \max_{\mathcal{Q}_{R^2}^i}\|x - C_0\| < \max_{ \mathcal{Q}_{R^2}^{i+1}} \|x - C_0\| \hspace{0.5cm} \forall i \in \mathbb{Z}  \nonumber \\
&R \geq R_0 \Rightarrow \mathcal{Q}_{R^2}^{i+1} \subseteq \mathcal{Q}_{R^2}^i \hspace{0.5cm} \forall i \in \mathbb{Z} \nonumber \\
\end{align}
\end{rem}
As such, from (\ref{E57c}) we give the following equation for $R_0 $ for any $i \in \mathbb{Z}$ and $p \in \mathbb{N}\setminus \{0\}$. 
\begin{align}\label{E58c}
R_0 = \max_{x \in \mathcal{Q}} \|x - C_0\| = \min \left\{R > 0 \biggr| \frac{\text{Vol}\left( \mathcal{Q}_{R^2}^{i+p} \cap \mathcal{Q}_{R^2}^i\right)}{\text{Vol}\left( \mathcal{Q}_{R^2}^{i+p}\right)} = 1 \right\} 
\end{align} showing that the exact computation of the volume of the intersection of balls is NP-Hard. 

The equation (\ref{E58c}) allows an approximation of $R_0$ because there are polynomial complexity randomized algorithms which can compute the volume of convex bodies, see \cite{convol1}, \cite{convol2}, \cite{convol3} and references therein. 
As the limit case, one can take $\mathcal{Q}_{R^2}^{i+p} = \mathcal{Q}_{R^2}^{\infty} $ which is a polytope and $\mathcal{Q}_{R^2}^{i} = \mathcal{Q}_{R^2}^{-\infty}$ which is the ball $\mathcal{B}(C_0,R)$.

Let $\mathcal{V}(R) := \frac{\text{Vol}\left( \mathcal{Q}_{R^2}^{i+p} \cap \mathcal{Q}_{R^2}^i\right)}{\text{Vol}\left( \mathcal{Q}_{R^2}^{i+p}\right)}$ then $\mathcal{V}(R) = 1$ for all $R \geq R_0$ because in this case $\mathcal{Q}_{R^2}^{i+p} \subseteq \mathcal{Q}_{R^2}^i$ for $p > 0$ according to Remark \ref{R9}. For $R < R_0$ thou, $\mathcal{Q}_{R^2}^{i+p} \not\subseteq \mathcal{Q}_{R^2}^i$  and $\mathcal{Q}_{R^2}^{i+p} $ will have some vertices outside $\mathcal{Q}_{R^2}^{i} $ and with them some volume. 
 
\subsection{A method for approximating $R_0$}

As already said, for $R< R_0$ one has 
$\mathcal{Q}_{R_0^2}^i \subseteq \mathcal{Q}_{R^2}^i$ if $i > 0$.

That is, the max indicator intersection of balls, reaches the ball $\mathcal{B}(C_0,R)$ for $i \to -\infty$, while the max indicator intersection of balls increases (includes the previous one) as $R$ decreases from $R_0$ if $i>0$ fixed. They reach the polytope $\mathcal{Q}_{R^2}^{\infty}$ for if $i \to \infty$.

 This asymmetry will be used in this subsection to provide a randomized 
method to approximate $R_0$. 

For $i \in \mathbb{Z}$, let us denote 
\begin{align}
\mathcal{S}_{R^2}^{i} = \partial \mathcal{B}(C_0,R) \cap \mathcal{Q}_{R^2}^{i}
\end{align} that is the surface of the boundary of $\mathcal{Q}_{R^2}^{-\infty} = \mathcal{B}(C_0,R)$ that is the set $\mathcal{Q}_{R^2}^i$. It is known that for $R = R_0$ one has $\mathcal{S}_{R^2}^{i} = \mathop{\text{argmax}}_{x \in \mathcal{Q}} \|x - C_0\|$ for $R>R_0$ one has $\mathcal{S}_{R^2}^{i} = \emptyset$. For the case $R < R_0$ we have the following result.

 Let us define
\begin{align}
\mathcal{R}_i(R) = \frac{|\mathcal{S}_{R^2}^{i}|}{|\partial \mathcal{B}(C_0,R)|} 
\end{align} where by $|\partial \mathcal{B}(C_0,R)|$ we denote the surface area of the ball $\mathcal{B}(C_0,R)$ and therefore by $|\mathcal{S}_{R^2}^{i}|$ we denote the surface area of $\mathcal{B}(C_0,R)$ which is in $\mathcal{Q}_{R^2}^{i}$. 
 
\textbf{First we give a negative result for $i >0$} assuming $C_0 \in \mathcal{Q}$. We start with the following lemma:
\begin{lemma}\label{L2.3} Let $i \geq 0$ then for $R < R_0$, the following holds:
\begin{align}
\lim_{n \to \infty} \frac{\mathcal{R}_i\left(R-\frac{\alpha}{n} \right)}{ \mathcal{R}_i(R)} \geq e^{\frac{\alpha}{R}}
\end{align} where $n$ is the dimension of the space and $\alpha \geq 0$. 
\end{lemma}
Thinking at $\mathcal{S}_{R^2}^i$ as the reunion of "holes" in the surface of $\mathcal{B}(C_0,R)$ one understands that this is null for $R>R_0$, contains a finite number of points for $R = R_0$ (i.e. the maximzers) and for $R < R_0$ the above lemma states that the quotient between the surface of the "holes" and that of the ball grows exponentially as $R$ departs from $R_0$.  

\begin{proof}

For $R < R_0$, we evaluate 
\begin{align}
\frac{\mathcal{R}_i\left(R-\frac{\alpha}{n} \right)}{ \mathcal{R}_i(R)} = \frac{\left|\mathcal{S}_{\left (R-\frac{\alpha}{n} \right)^2}^i \right|}{|\mathcal{S}_{ R^2}^i|} \cdot \frac{|\partial \mathcal{B}\left(C_0,R \right)|}{|\partial \mathcal{B}\left(C_0,R-\frac{\alpha}{n} \right)|} 
\end{align} For the last term, it is obvious that 
\begin{align}
\frac{|\partial \mathcal{B}\left(C_0,R \right)|}{|\partial \mathcal{B}\left(C_0,R-\frac{\alpha}{n} \right)|} = \left( \frac{R}{R - \frac{\alpha}{n}} \right)^{n-1} = \left( \frac{1}{1 - \frac{\alpha}{R} \cdot \frac{1}{n}}\right)^{n-1}
\end{align} while for the first term in the product, since $\mathcal{S}_{ R^2}^i \subseteq \mathcal{S}_{\left (R-\frac{\alpha}{n} \right)^2}^i $ follows that 
\begin{align}
\frac{\left|\mathcal{S}_{\left (R-\frac{\alpha}{n} \right)^2}^i \right|}{|\mathcal{S}_{ R^2}^i|} \geq 1  \Rightarrow \frac{\mathcal{R}_i\left(R-\frac{\alpha}{n} \right)}{ \mathcal{R}_i(R)} \geq \left( \frac{1}{1 - \frac{\alpha}{R} \cdot \frac{1}{n}}\right)^{n-1}
\end{align}

Finally, evaluate the term

\begin{align}
\left( \frac{1}{1 - \frac{\alpha}{R} \cdot \frac{1}{n}}\right)^{n-1}  = \left(\left(1 - \frac{\alpha}{R} \cdot \frac{1}{n}\right)^{n\cdot \frac{R}{\alpha}}\right)^{\frac{-\alpha}{R} \cdot \frac{n-1}{n}} \to^{n \to \infty} e^{\frac{\alpha}{R}} 
\end{align}

\end{proof}
%
%
%
%
%

Taking a point $\hat{x}$ on the surface of $\mathcal{B}(C_0,R)$, the quantity  $\mathcal{R}_i(R)$ gives the probability of  having $\hat{x} \in \mathcal{Q}_{R^2}^i$. For $R \geq R_0$ this probability is zero, but as $R$ drops below $R_0$ Lemma \ref{L2.3} shows that this increases exponentially with $\alpha$ where $\alpha = n\cdot (R_0 - R)$.  

We propose the following algorithm for computing $R_0$.

\begin{algorithm} 
\caption{Procedure B}\label{procB}
 This procedure gives a method to approximate $R_0$. 
\begin{algorithmic}[1]
\Require $R < R_0$ and $i \in \mathbb{Z}$ and $N > 0$
 
\State $\mathcal{D} \gets \{ \hat{x}_1, \hdots, \hat{x}_N\}$ with $\hat{x}_k \in \partial \mathcal{B}(C_0,R)$ with uniform distribution. 
\While{$\mathcal{D} \cap \mathcal{Q}_{R^2}^i \neq \emptyset $}
\State Increase $R$
\State $\mathcal{D} \gets \{ \hat{x}_1, \hdots, \hat{x}_N\}$ with $\hat{x}_k \in \partial \mathcal{B}(C_0,R)$ with uniform distribution. 
\EndWhile

\end{algorithmic}

\end{algorithm}

From Lemma \ref{L2.3} follows that for any $\epsilon > 0$ exists $ n_{\epsilon}$ such that for all $n \geq n_{\epsilon}$
\begin{align}
\frac{\mathcal{R}_i\left( R - \frac{\alpha}{n}\right)}{\mathcal{R}_i(R)} \geq e^{\frac{\alpha}{R}} - \epsilon
\end{align} For $\beta \in (0,1)$, let $\alpha = \beta\cdot R \cdot n$ then 
\begin{align}\label{E67c}
\mathcal{R}_i(R) \leq \frac{1}{e^{\beta \cdot n} - \epsilon} \cdot \mathcal{R}_i(R - \beta \cdot R) 
\end{align} Note that exists $\beta \in (0,1)$ such that $\mathcal{R}_i((1-\beta) \cdot R) = 1$. Indeed, for 
\begin{align} \mathcal{R}_i( R) = \frac{|\mathcal{S}_{R^2}^{i}|}{|\partial \mathcal{B}(C_0,R)|} = \frac{|\mathcal{Q}_{R^2}^{i} \cap \partial \mathcal{B}(C_0,R)|}{|\partial \mathcal{B}(C_0,R)|}
\end{align} note that $\mathcal{Q} \subseteq \mathcal{Q}_{0^2}^i$ and as $R$ decreases $\mathcal{B}(C_0,R) \subseteq \mathcal{Q}$ eventually because $C_0 \in \mathcal{Q}$. From  (\ref{E67c}) follows that 

\begin{align}
\mathcal{R}_i(R) \leq \frac{1}{e^{\beta \cdot n} - \epsilon} 
\end{align} hence applying Algorithm \ref{procB} is hopeless because the probability of randomly taking a point on the surface of $\mathcal{B}(C_0,R)$ in $\mathcal{Q}_{R^2}^i$ is exponentially low. 

\textbf{Next, we motivate applying the Algorithm \ref{procB} for $i < 0$.} We recall from the previous sections that for $i > 0$ and $R>0$, the radii of the intersecting balls forming $\mathcal{Q}_{R^2}^{-i}$ are $r_{k,-i}$ with, see (\ref{E54c}) 

\begin{align}\label{E70c}
r_{k,-i}^2 - R^2 &= (1 - \lambda)^i\cdot \left(r_{k,0}^2 - R^2 - \left(1 - (1-\lambda)^i \right) \cdot \|C_0 - C_{k,0}\|^2 \right)
\end{align} and the centers are, see (\ref{E46c})

\begin{align}\label{E71c}
C_{k,-i} = (1-\lambda)^i \cdot C_{k,0} + \left( 1 - (1 - \lambda)^i \right)\cdot C_0
\end{align}

Consider the following geometry problem: a ball $\mathcal{B}(D_0,\rho)$ and $m$ other balls $\mathcal{B}(D_k,\rho_k)$ with $k \in \{1, \hdots, m\}$. Define the function:

\begin{align}\label{E73c}
f_k(\rho) = \frac{\rho_k - \rho}{\|D_k - D_0\|} 
\end{align} Note that if $f_k(\rho) \geq 1$ for all $k$ then $\mathcal{B}(D_0,\rho) \subseteq \bigcap_{k=1}^m\mathcal{B}(D_k,\rho_k)$, while if $f_k(\rho) \leq -1$ for all $k$ then $\bigcap_{k=1}^m\mathcal{B}(D_k,\rho_k) \subseteq \mathcal{B}(D_0,\rho) $


As such we note that the larger the value of $f_k(\rho)$ (for each $k$) is, the larger the area of $\partial \mathcal{B}(D_0,\rho)$ in the set $\bigcap_{k=1}^m\mathcal{B}(D_k,\rho_k)$.

We investigate:

\begin{align}\label{E72c}
\frac{r_{k,-i}^2 - R^2}{\|C_{k,-i} - C_0\|} = \frac{r_{k,0}^2 - R^2 - \left(1 - (1-\lambda)^i \right) \cdot \|C_0 - C_{k,0}\|^2 }{ \| C_{k,0} - C_0\|}
\end{align}  

On the above lines, we give our final lemma:

\begin{lemma}\label{L2.4}
For $k \in \{1, \hdots, m\}$ and $|i|$ large enough, if $r_{k,0}^2 - R_0^2 - \|C_{k,0} - C_0\|^2 \geq 0$  then for all $0 < R \leq R_0$ one has 
\begin{align}
 \frac{r_{k,-i}^2 - R^2}{\|C_{k,-i} - C_0\|}  <   \frac{r_{k,-i-1}^2 - R^2}{\|C_{k,-i-1} - C_0\| } 
\end{align}

\end{lemma}
The lemma shows that for a fixed $R$, the value of $\frac{r_{k,-i}^2 - R^2}{\|C_{k,-i} - C_0\|}$, for a given $k$, increases as $|i|$ increases. 
\begin{proof}
 Indeed, from (\ref{E72c}) one has:

\begin{align}
\frac{\frac{r_{k,-i}^2 - R^2}{\|C_{k,-i} - C_0\|}}{ \frac{r_{k,-i-1}^2 - R^2}{\|C_{k,-i-1} - C_0\| }} &= \frac{\frac{r_{k,0}^2 - R^2 - \left(1 - (1-\lambda)^i \right) \cdot \|C_0 - C_{k,0}\|^2 }{ \| C_{k,0} - C_0\| }}{\frac{r_{k,0}^2 - R^2 - \left(1 - (1-\lambda)^{i+1} \right) \cdot \|C_0 - C_{k,0}\|^2 }{ \| C_{k,0} - C_0\| }} \nonumber \\
&= \frac{(1-\lambda)^i \cdot \| C_{k,0} - C_0\| + b_k}{(1-\lambda)^{i+1} \cdot \| C_{k,0} - C_0\| + b_k} 
\end{align} where $b_k = \frac{r_{k,0}^2 - R^2 - \|C_{k,0} - C_0\|^2}{\|C_{k,0} - C_0\|} $

For the problems where $r_{k,0}^2 - R_0^2 - \|C_{k,0} - C_0\|^2 \leq 0$ follows that $b_k \leq 0$ hence it is obvious that 
\begin{align}
 \frac{(1-\lambda)^i \cdot \| C_{k,0} - C_0\| + b_k}{(1-\lambda)^{i+1} \cdot \| C_{k,0} - C_0\| + b_k} < 1
\end{align} for large enough $|i|$, because $(1-\lambda)^i \geq (1-\lambda)^{i+1}$ since $(1-\lambda) \in (0,1)$.

From (\ref{E70c}) follows that if $r_{k,-i}^2 \leq R^2$ then so is $r_{k,-i-1}^2$ hence 
\begin{align}
\frac{\frac{r_{k,-i}^2 - R^2}{\|C_{k,-i} - C_0\|}}{ \frac{r_{k,-i-1}^2 - R^2}{\|C_{k,-i-1} - C_0\| }} = \frac{  \left|\frac{r_{k,-i}^2 - R^2}{\|C_{k,-i} - C_0\|} \right|}{  \left|\frac{r_{k,-i-1}^2 - R^2}{\|C_{k,-i-1} - C_0\| } \right|}
\end{align} hence the lemmas conclusion follows. 

\end{proof}

We end the section with the following remarks:
\begin{rem}
Although the quantity in the Lemma \ref{L2.4} although is not exactly the same as in (\ref{E73c}) it still shows an improvement in the directions related to (\ref{E73c}), since $r_{k,-i}^2 - R^2 = (r_{k,-i} - R)\cdot (r_{k,-i} + R)$. 


\end{rem}

\section{Numerical Results}
In this section we show the application of Theorem \ref{T1} for an intersection of balls. In Figure \ref{fig1} and the close-up Figure \ref{fig2} one sees the intersection of balls with green, the max indicators ball-polyedra with blue for different values of $R$. With magenta is the intersection of balls which first enters the intersection of balls. As predicted by the presented theory, the parameter $R$ for which this happens is the maximum distance to the given point over the green intersection of balls. 

	 \begin{figure}[h]
  \includegraphics[width = 9cm]{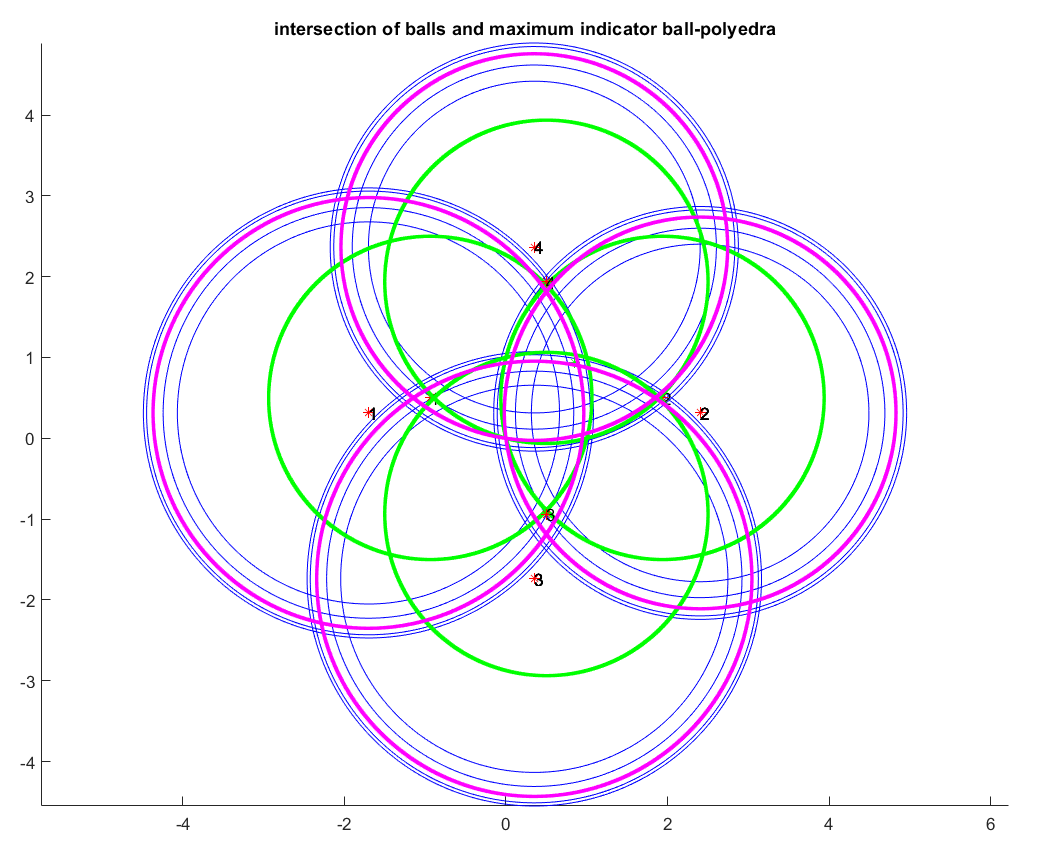}
  \caption{ The given intersection of balls with green, few instances of the family of max indicator ball polyedra with blue, the member of the family corresponding to the maximum distance }
  \label{fig1}
\end{figure}

 \begin{figure}[h]
  \includegraphics[width = 9cm]{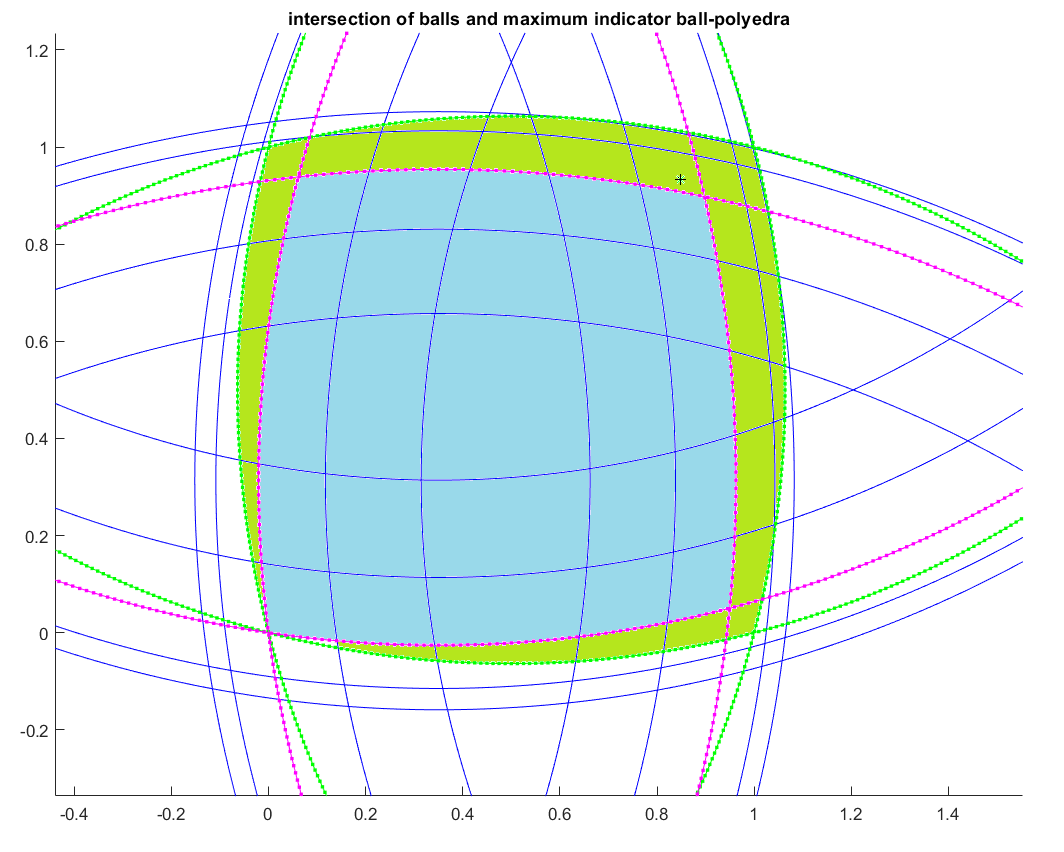}
  \caption{ Close-up and enhancement of Figure \ref{fig1}. The given intersection of balls with green filling and green boundary, few instances of the family of max indicator ball polyedra with blue boundary and the member of the family corresponding to the maximum distance with cyan filling and magenta boundary. The point $C_0$ is the black cross. }
  \label{fig2}
\end{figure}

In Figure \ref{fig4} and its close-up Figure \ref{fig5} one can see the ball polyedra obtained by applying Procedure A to the given intersection of balls. The initial intersection of balls is depicted with green. After one step of Procedure A, the centers of the blue balls are obtained. We plot them knowing the correct $R_0$. Successive applications result in successively farther centers and larger radii. It is visible that $C_0$ remains in the convex hull of the balls centers. 
   \begin{figure}[h]
  \includegraphics[width = 9cm]{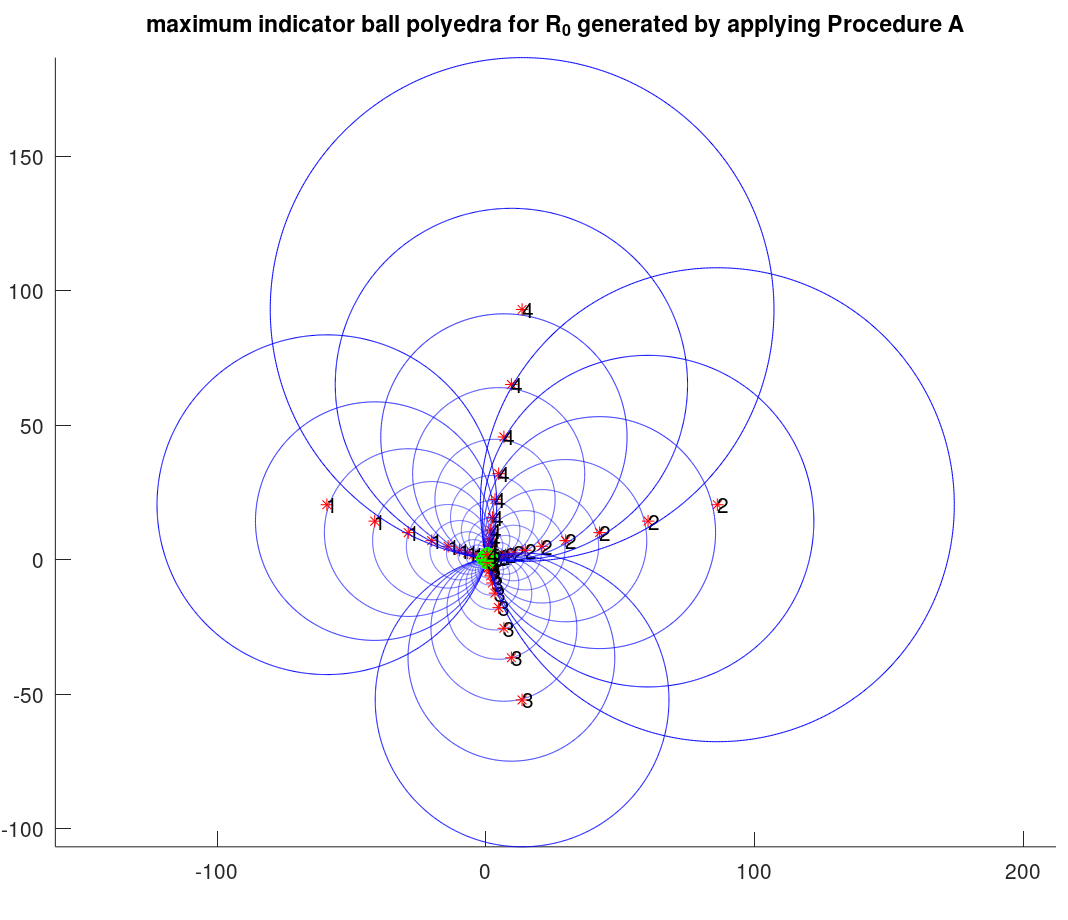}
  \caption{ The given intersection of balls with green, the centers generated by Procedure A with red. These can be obtained without the knowledge of the max distance $R_0$, although we plot the circles with the correct radius. }
  \label{fig3}
\end{figure}

  \begin{figure}[h]
  \includegraphics[width = 9cm]{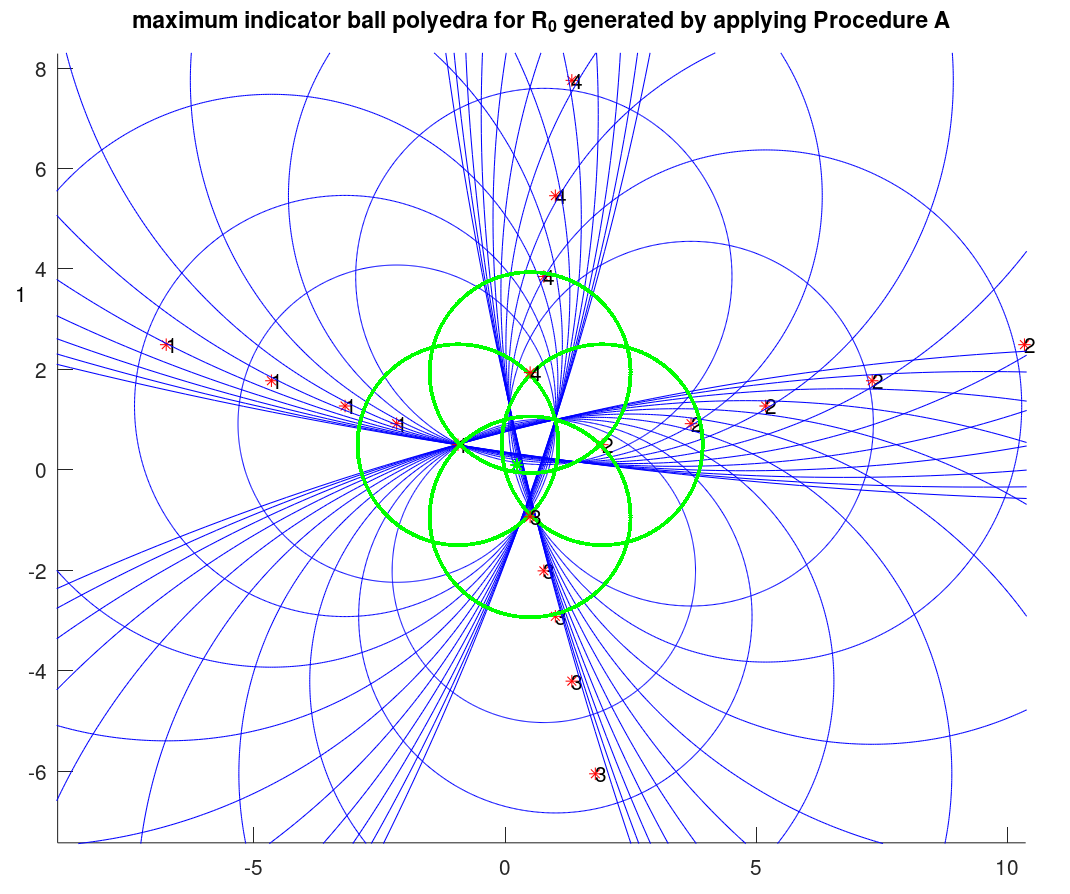}
  \caption{ Close-up of Figure \ref{fig3}. The formed polytope is $\mathcal{Q}_{R_0}^{\infty}$}
  \label{fig4}
\end{figure}

In in Figure \ref{fig5} and its close-up Figure \ref{fig6} one can see intersection of balls (with blue) to which if one applies Procedure A the given intersection of balls (with green) is obtained. As a limit case these are all generated form the smallest ball (with magenta) centered in $C_0$ enclosing the initial intersection of balls. Note the validity of Remark \ref{R6}. 
 \begin{figure}[h]
  \includegraphics[width = 9cm]{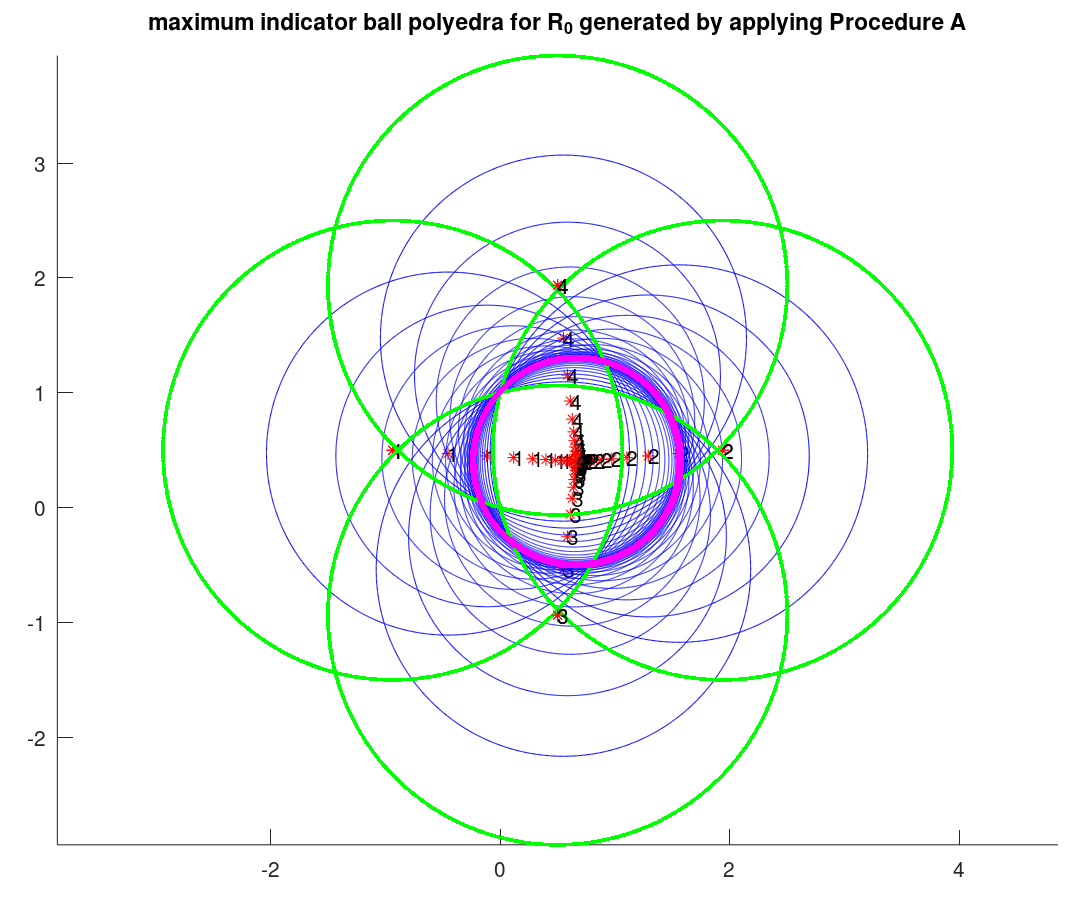}
  \caption{ The given intersection of balls with green, the limit case with magenta and with blue intersection of balls which evolve under Procedure A into the green intersection of balls. Note the confirmation of Remark \ref{R6} }
  \label{fig5}
\end{figure}

  \begin{figure}[h]
  \includegraphics[width = 9cm]{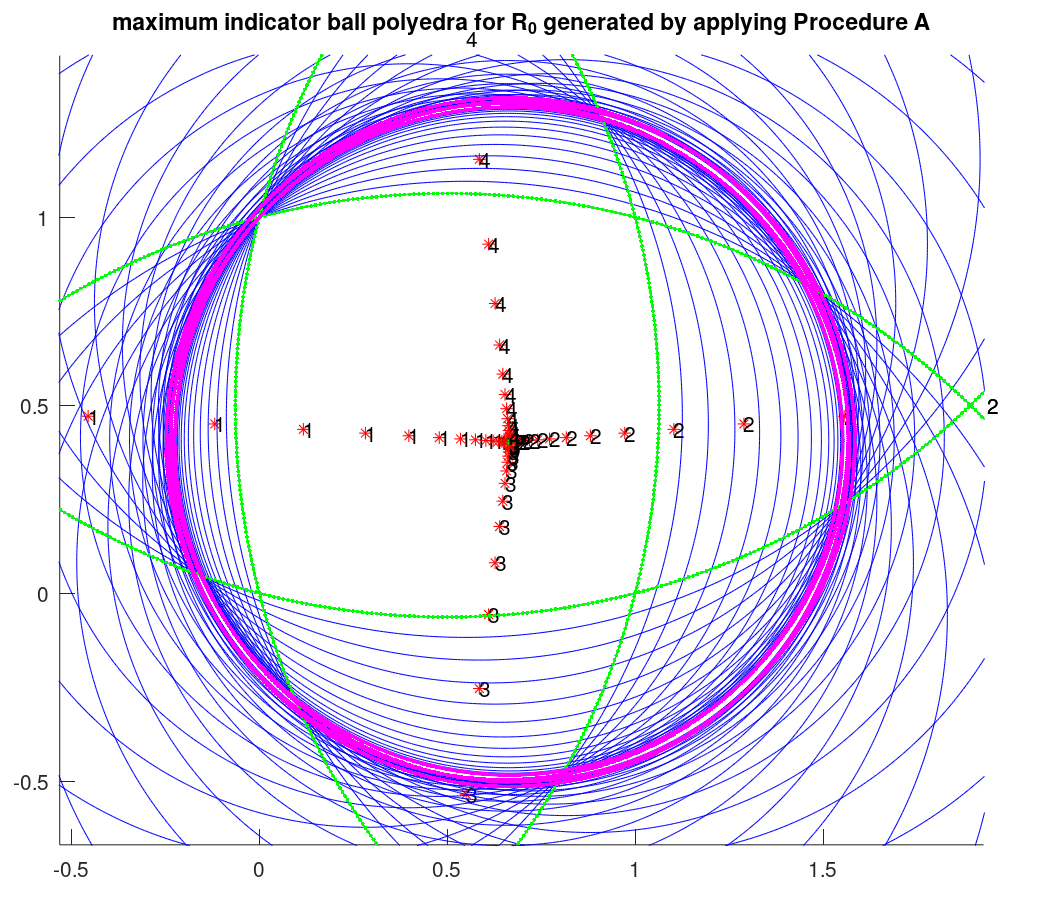}
  \caption{ Close-up of Figure \ref{fig5}. With magenta are the last balls (which have the centers closest to $C_0$) The obtained ball is $\mathcal{Q}_{R_0}^{-\infty}$ }
  \label{fig6}
\end{figure}
%
%

\section{Conclusion and future work}

We presented in this paper a method for approximating the maximum distance $R_0$ over an intersection of balls $\mathcal{Q}$ to a given point $C_0$. Starting from the given problem we construct a sequence of larger intersection of balls which preserve the maximizers to the given problem, then propose an algorithm which gives an upper bound to the maximum distance by maximizing over a particulary constructed, larger intersection of balls. The later problem is shown to allow a polynomial algorithm to obtain the solution. Then we also provide a method to compute a lower bound for the maximum distance.

From the proof of the Theorem \ref{T1}, as future work, we propose using the presented framework to analyze the problem: $\max_{x\in \mathcal{P}} \|x - C_0\|$ where $C_0 \in \mathbb{R}^n$ and $\mathcal{P}$ is a polytope. Assume 
\begin{align}
\mathcal{P} = \left\{ x | \begin{bmatrix} A_1^T\\ \vdots\\ A_m^T \end{bmatrix} \cdot x + \begin{bmatrix} b_1\\ \vdots \\ b_m\end{bmatrix} \succeq 0_{m\times 1}\right\}
\end{align} with $A,B$ matrices of appropriate size. Let $h(x) = t\cdot \sum_{k=1}^m \log\left( \frac{1}{A_k^T\cdot x + b_k}\right)$ and $g(x) = \log\left( \|x - C_0\|^2 + 1\right)$ with $t > 0$. Consider the optimization problem:
\begin{align}
\max_{x \in \{x | h(x) \leq 0\}} g(x)
\end{align} For this, as in the proof of Theorem \ref{T1} consider the function $h(x) - g(x)$. We can choose $t$ large enough such that $h-g$ is convex, see Figure \ref{fig8}. Letting $\mathcal{H}^{\star} = \min_{h(x)\leq 1} h(x) - g(x)$ for $\mathcal{H}^{\star} \in \{x | h(x)\leq 0\}$ one should study if a similar argument to the theorem proof can be made which says that the maximum of the function $g$ over $h(x)\leq 0$ is the smallest $R>0$ for which the set $\mathcal{P}_{R^2}$ enters the set $\{x | h(x) \leq 0\}$, where $\mathcal{P}_{R^2} = \{x | h(x) - g(x) \leq -\log(1 + R^2)\}$. To be analyzed as well if $\mathcal{P}_{R^2} \subseteq \{x | h(x) \leq 0\}$ can be characterized by $\mathcal{P}_{R^2} \subseteq \mathcal{P}$. 

 \begin{figure}[h]
  \includegraphics[width = 9cm]{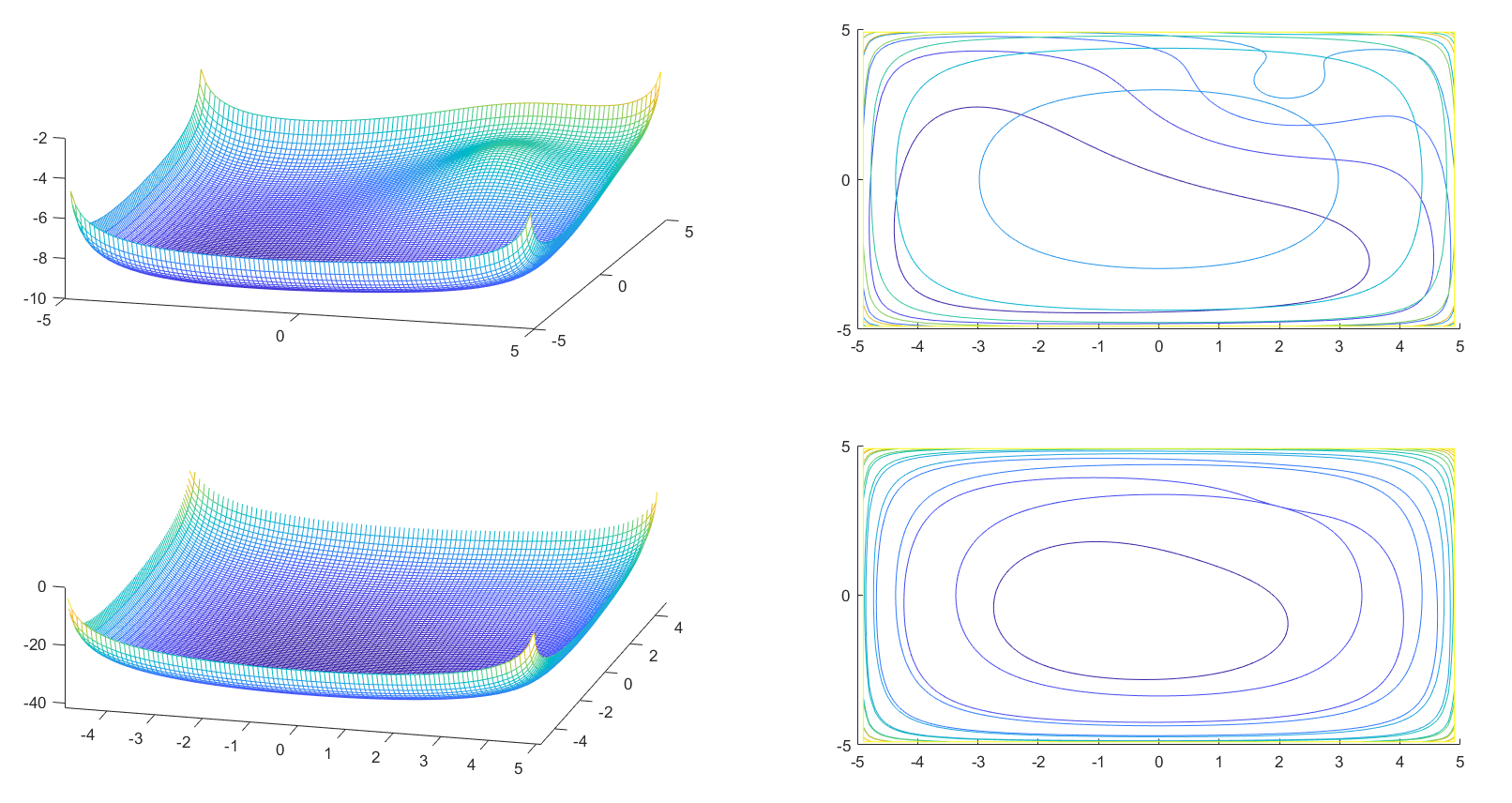}
  \caption{ The function $h(x) - g(x)$ for $t = 1$ then the same function for $t \gg 1$ for some example function. It shows that the function can be made convex.}
  \label{fig8}
\end{figure}

%


\end{document}